\documentclass[conference,twocolumn,final,letterpaper]{IEEEtran}%
\usepackage{amsfonts}
\usepackage{amsmath}
\usepackage{amssymb}
\usepackage[dvips]{graphicx}
\usepackage{cite}%
\newtheorem{theorem}{Theorem}

\newtheorem{corollary}[theorem]{Corollary}

\newtheorem{definition}[theorem]{Definition}

\newtheorem{remark}[theorem]{Remark}

\begin{document}

\title{Sum-Capacity of Ergodic Fading Interference and Compound Multiaccess Channels}
\pubid{\ \ }
\specialpapernotice{\ \ }%

\author{\authorblockN{Lalitha Sankar\authorrefmark{1}%
, Elza Erkip\authorrefmark{1}\authorrefmark{2}, H. Vincent Poor\authorrefmark
{1}}
\authorblockA{\authorrefmark{1}Dept. of Electrical Engineering,
Princeton University,
Princeton, NJ 08544.
{lalitha,eerkip,poor}@princeton.edu\\}
\authorblockA{\authorrefmark{2}Dept. of Electrical and Computer Engineering,
Polytechnic University,
Brooklyn, NY 11201.
elza@poly.edu\\}}%
%

\maketitle
%

\begin{abstract}%

\footnotetext{This research is supported in part by the National\ Science
Foundation under Grants ANI-03-38807 and CNS-06-25637 and in part by a
fellowship from the Princeton University Council on\ Science and
Technology.}The problem of resource allocation is studied for two-sender
two-receiver fading Gaussian interference channels (IFCs) and compound
multiaccess channels (C-MACs). The senders in an IFC communicate with their
own receiver (unicast) while those in a C-MAC communicate with both receivers
(multicast). The instantaneous fading state between every transmit-receive
pair in this network is assumed to be known at all transmitters and receivers.
Under an average power constraint at each source, the sum-capacity of the
C-MAC and the power policy that achieves this capacity is developed. The
conditions defining the classes of \textit{strong }and \textit{very strong
ergodic }IFCs are presented and the multicast sum-capacity is shown to be
tight for both classes.%

\end{abstract}%

\section{Introduction}

The two-user interference channel (IFC) and the two-user compound multiaccess
channel (C-MAC)\ model networks with two sources (senders or transmitters) and
two destinations (receivers). The unicast case in which the message from each
source is intended for only one destination is modeled as an IFC while the
multicast case in which both messages are intended for both destinations is
modeled as a C-MAC (see Fig. \ref{Fig_IC}). The capacity region of a discrete
memoryless C-MAC is obtained in \cite{cap_theorems:Ahlswede_CMAC}. The
capacity region of both the discrete memoryless and the Gaussian IFC remain
open problems; however, for certain classes of time-invariant IFCs satisfying
specific well-defined constraints the capacity region is known (see for e.g.,
\cite{cap_theorems:Sato_IC,cap_theorems:Carleial_VSIFC,cap_theorems:CostaElGamal_IC,cap_theorems:MotaKhan}
and the references therein).

\bigskip

The ergodic sum-capacity and the capacity region of a multiaccess channel
(MAC)\ are studied in \cite{cap_theorems:Knopp_Humblet} and
\cite{cap_theorems:TH01}, respectively, under the assumption that the channel
states and statistics are known at all nodes. These papers also develop the
rate-optimal power policies. The ergodic capacity of a C-MAC, however, is not
a straightforward extension of these results. For a parallel Gaussian IFC, in
\cite{cap_theorems:YuCioffiGinnis}, the authors propose a sub-optimal
iterative water-filling solution when every receiver views signals from the
unintended transmitters as interference. In \cite{cap_theorems:ChuCioffi_IC},
the capacity of a parallel Gaussian IFC where every parallel subchannel is
strong, i.e., its signal-to-noise and interference-to-noise ratios at each
receiver satisfy specific conditions \cite{cap_theorems:Sato_IC}, is
developed. In this paper, we study the problem of resource allocation for the
two-user ergodic fading IFC and C-MAC under the assumption that the
instantaneous fading state between each transmit-receive pair in this network
is known at all transmitters and receivers. We develop the ergodic
sum-capacity for the C-MAC which in turn lower bounds the sum-capacity of the
IFC. We further develop the conditions defining the classes of \textit{strong
}and \textit{very strong ergodic} IFCs with resource allocation and show that
the C-MAC lower bounds are tight for both classes. Our work differs from
\cite{cap_theorems:ChuCioffi_IC} in that we develop capacity results for an
ergodic IFC that is strong or very strong on average, i.e., the constraints
for these classes require averaging over all channel instantiations, and thus,
our result subsumes that in \cite{cap_theorems:ChuCioffi_IC}.

\bigskip

The sum-capacity optimal policy for the C-MAC is motivated by the work in
\cite{cap_theorems:SankarLiang_Conf} on maximizing the sum-rate of an ergodic
fading two-user orthogonal multiaccess relay channel (MARC)
\cite{cap_theorems:SankarLiang_Conf} when the relay employs a
decode-and-forward (DF) strategy. For the MARC, a DF relay acts as a decoding
receiver; this enables us to generalize from
\cite{cap_theorems:SankarLiang_Conf} that when both receivers in a two-sender
two-receiver network decode messages from both sources (users), the resulting
sum-rate belongs to one of five disjoint cases or lies on the boundary of any
two of them (boundary cases). Further, the sum-rate optimal policy either: 1)
exploits the multiuser fading diversity to opportunistically schedule users
analogous to the fading MAC
\cite{cap_theorems:Knopp_Humblet,cap_theorems:TH01} or 2) involves
simultaneous water-filling over two independent point-to-point links. We first
develop the capacity region of the ergodic C-MAC; the resulting region is
shown to lie within the capacity region of an ergodic IFC. The sum-rate
optimal policy described above achieves the C-MAC sum-capacity and a lower
bound on the IFC sum-capacity. We develop the conditions for the very strong
IFC and show that the C-MAC lower bound for one of the five disjoint cases is
tight for this IFC class. We define the conditions for the strong ergodic IFC
and prove that when these conditions are met, the IFC sum-capacity is the
C-MAC sum-capacity for one of three other disjoint cases or three boundary
cases. We also show that, in contrast to the non-fading case
\cite{cap_theorems:Sato_IC}, the constraints for both classes of IFC depend on
both the channel statistics and average power constraints.

\bigskip

The paper is organized as follows. In Section \ref{Section 2}, we model the
ergodic fading Gaussian C-MAC and IFC. In Section \ref{section 3} we present
the C-MAC capacity region and determine the power policies that maximize the
sum-capacity. In Section \ref{Sec_4}, we define the strong and very strong
ergodic IFC conditions and show that the sum-capacity in Section
\ref{section 3} is tight when the conditions for either the strong or the very
strong ergodic IFC hold.

\section{\label{Section 2}Channel Model and Preliminaries}%

\begin{figure}
[ptb]
\begin{center}
\includegraphics[
height=1.9233in,
width=3in
]%
{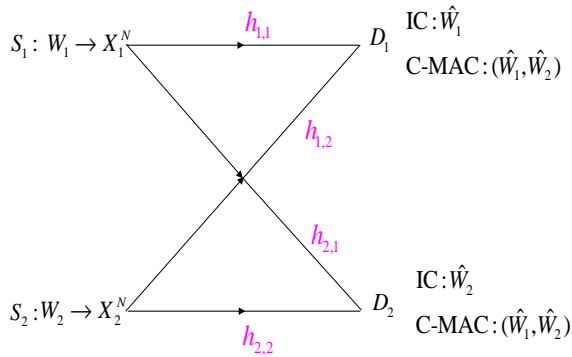}%
\caption{The two-user Gaussian\ IFC or C-MAC.}%
\label{Fig_IC}%
\end{center}
\end{figure}

A two-sender two-receiver Gaussian IFC consists of two source nodes $S_{1}$
and $S_{2}$, and two destination nodes $D_{1}$ and $D_{2}$ as shown in Fig.
\ref{Fig_IC}. Source $S_{k}$, $k=1,2$, uses the channel $N$ times to transmit
its messages $W_{k}$, distributed uniformly in the set $\,\{1,2,\ldots
,2^{B_{k}}\}$, to its intended receiver, $D_{k}$, at a rate $R_{k}=B_{k}/N$
bits per channel use. In each use of the channel, $S_{k}$ transmits the signal
$X_{k}$ while the destinations $D_{1}$ and $D_{2}$ receive $Y_{1}$ and $Y_{2}%
$, respectively, such that
\begin{align}
Y_{1}  &  =H_{1,1}X_{1}+H_{1,2}X_{2}+Z_{1}\label{IC_Y1}\\
Y_{2}  &  =H_{2,1}X_{1}+H_{2,2}X_{2}+Z_{2} \label{IC_Y2}%
\end{align}
where $Z_{1}$ and $Z_{2}$ are independent circularly symmetric complex
Gaussian noise random variables with zero means and unit variances. For the
special case when the messages at $S_{1}$ and $S_{2}$ are intended for both
destinations, the model defined by (\ref{IC_Y1}) and (\ref{IC_Y2}) results in
a two-user Gaussian C-MAC (see Fig. \ref{Fig_IC}). We write $\mathbf{H}$ to
denote the random matrix of fading states, $H_{k,m}$, for all $k,m=1,2$, such
that $\mathbf{H}$ is a realization for a given channel use of a jointly
stationary and ergodic (not necessarily Gaussian) fading process
\underline{$\mathbf{H}$}. Note that $H_{k,m}$ for all $k,m$, are not assumed
to be independent. We also assume that over $N$ uses of the channel, the
source transmissions are constrained in power according to%
\begin{equation}
\left.  \sum\limits_{i=1}^{N}\left\vert X_{k,i}\right\vert ^{2}\leq
N\overline{P}_{k}\right.  \text{ for all }k=1,2 \label{Avg_pwr_constraint}%
\end{equation}
where $X_{k,i}$ denotes the transmitted signal from source $k$ in the $i^{th}$
channel use. Since the sources know the fading states of the links on which
they transmit, they can allocate their transmitted signal power according to
the channel state information. We write $P_{k}(\mathbf{H})$ to denote the
power allocated at the $k^{th}$ transmitter as a function of the channel
states $\mathbf{H}$. For an ergodic fading channel, (\ref{Avg_pwr_constraint})
then simplifies to
\begin{equation}
\left.  \mathbb{E}\left[  P_{k}(\mathbf{H})\right]  \leq\overline{P}%
_{k},\right.  \text{ }k=1,2, \label{GIC_pwr_constraint}%
\end{equation}
where the expectation in (\ref{GIC_pwr_constraint}) is over the joint
distribution of $\mathbf{H}$. We write $\underline{P}\left(  \mathbf{H}%
\right)  $ to denote a vector of power allocations with entries $P_{k}%
(\mathbf{H})$, for all $k$, and define $\mathcal{P}$ to be the set of all
$\underline{P}\left(  \mathbf{H}\right)  $ whose entries satisfy
(\ref{GIC_pwr_constraint}). The capacity region $\mathcal{C}_{\text{IFC}}$
$\left(  \mathcal{C}_{\text{C-MAC}}\right)  $ of a two-user IFC\ (C-MAC) is
defined as the closure of the set of rate tuples $(R_{1},R_{2})$ such that the
destinations can decode their intended messages with an arbitrarily small
positive error probability $\epsilon$. For ease of notation, we henceforth
omit the functional dependence of $\underline{P}$ on $\mathbf{H}$. We write
random variables (e.g. $H_{k,j}$) with uppercase letters and their
realizations (e.g. $h_{k,j}$) with the corresponding lowercase letters. We
write $\mathcal{K}=\left\{  1,2\right\}  $ to denote the set of transmitters,
the notation $C(x)$ $=$ $\log(1+x)$ where the logarithm is to the base 2,
$\left(  x\right)  ^{+}=\max(x,0)$, and write $R_{\mathcal{S}}$ $=$ $%
{\textstyle\sum\nolimits_{k\in\mathcal{S}}}
R_{k}$ for any ${\mathcal{S}}$ $\subseteq\mathcal{K}$.

\section{\label{section 3}C-MAC: Sum-Capacity and Optimal Policy}

The capacity region of a two-transmitter (sender) two-receiver discrete
memoryless (d.m.) channel, now often referred to as a d.m. compound MAC, is
developed in \cite{cap_theorems:Ahlswede_CMAC}. For each choice of input
distribution at the two independent sources, this capacity region is an
intersection of the MAC capacity regions achieved at the two receivers. The
techniques in \cite{cap_theorems:Ahlswede_CMAC} can be easily extended to
develop the capacity region for a Gaussian C-MAC with fixed channel gains. For
the Gaussian C-MAC, one can show that Gaussian signaling achieves the capacity
region using the fact that Gaussian signaling maximizes the MAC region at each
receiver. Thus, the Gaussian C-MAC capacity region is an intersection of the
Gaussian\ MAC capacity regions achieved at $D_{1}$ and $D_{2}$. For a
stationary and ergodic process \underline{$\mathbf{H}$}, the channel in
(\ref{IC_Y1}) and (\ref{IC_Y2}) can be modeled as a set of parallel Gaussian
C-MACs, one for each fading instantiation $H$. For the ergodic fading case,
the capacity region $\mathcal{R}_{\text{C-MAC}}$, achieved over all
$\underline{P}\in\mathcal{P}$ is given by the following theorem.

\begin{theorem}
The capacity region, $\mathcal{C}_{\text{C-MAC}}$, of an ergodic fading
Gaussian\ C-MAC is%
\begin{equation}
\mathcal{C}_{\text{C-MAC}}=\bigcup_{\underline{P}\in\mathcal{P}}\left\{
\mathcal{C}_{1}\left(  \underline{P}\right)  \cap\mathcal{C}_{2}\left(
\underline{P}\right)  \right\}
\end{equation}
where for all $\mathcal{S}\subseteq\mathcal{K}$ and $j=1,2,$ we have%
\begin{equation}
\mathcal{C}_{j}\left(  \underline{P}\right)  =\left\{  \left(  R_{1}%
,R_{2}\right)  :R_{\mathcal{S}}\leq\mathbb{E}\left[  C\left(  \sum
_{k\in\mathcal{S}}\left\vert H_{j,k}\right\vert ^{2}P_{k}\right)  \right]
\right\}  . \label{CMAC_Cj}%
\end{equation}

\end{theorem}

\begin{proof}
The achievability follows from using Gaussian signaling and decoding at both
receivers. For the converse, we apply the proof techniques developed for the
capacity of an ergodic fading MAC in \cite{cap_theorems:TH01}. For any
$\underline{P}\in\mathcal{P}$, one can use limiting arguments (see for e.g.,
\cite[Appendix B]{cap_theorems:TH01}) to show that for asymptotically
error-free performance at receiver $j$, for all $j$, the achievable region has
to be bounded as%
\begin{equation}%
\begin{array}
[c]{cc}%
R_{\mathcal{S}}\leq\mathbb{E}\left[  \log\left(  1+\sum_{k\in\mathcal{S}%
}\left\vert H_{j,k}\right\vert ^{2}P_{k}\right)  \right]  , & j=1,2.
\end{array}
\end{equation}
The proof is completed by taking the union of the region over all
$\underline{P}\in\mathcal{P}$.
\end{proof}

\begin{corollary}
\label{Cor_1}The interference channel ergodic capacity region $\mathcal{C}%
_{\text{IFC}}$ is bounded as $\mathcal{C}_{\text{C-MAC}}\subseteq
\mathcal{C}_{\text{IFC}}$.
\end{corollary}

Corollary \ref{Cor_1} follows from the argument that a rate pair in
$\mathcal{C}_{\text{C-MAC}}$ is achievable for the IFC since $\mathcal{C}%
_{\text{C-MAC}}$ is the capacity region when both messages are decoded at both receivers.

\begin{remark}
The capacity region $\mathcal{C}_{\text{C-MAC}}$ is convex. This follows from
the convexity of the set $\mathcal{P}$ and the concavity of the $\log$ function.
\end{remark}

The capacity region $\mathcal{C}_{\text{C-MAC}}$ is a union of the
intersection of the pentagons $\mathcal{C}_{1}\left(  \underline{P}\right)  $
and $\mathcal{C}_{2}\left(  \underline{P}\right)  $ achieved at $D_{1}$ and
$D_{2},$ respectively, where the union is over all $\underline{P}%
\in\mathcal{P}$. The region $\mathcal{C}_{\text{C-MAC}}$ is convex, and thus,
each point on the boundary of $\mathcal{C}_{\text{C-MAC}}$ is obtained by
maximizing the weighted sum $\mu_{1}R_{1}$ $+$ $\mu_{2}R_{2}$ over all
$\underline{P}\in\mathcal{P}$, and for all $\mu_{1}>0$, $\mu_{2}>0$.
Specifically, we determine the optimal policy $\underline{P}^{\ast}$ that
maximizes the sum-rate $R_{1}+R_{2}$ when $\mu_{1}$ $=$ $\mu_{2}$ $=$ $1$.
Using the fact that the rate regions $\mathcal{C}_{1}\left(  \underline
{P}\right)  $ and $\mathcal{C}_{2}\left(  \underline{P}\right)  $ are
pentagons, in Figs. \ref{Fig_Case12} and \ref{Fig_Case3abc} we illustrate the
five possible choices for the sum-rate resulting from an intersection of
$\mathcal{C}_{1}\left(  \underline{P}\right)  $ and $\mathcal{C}_{2}\left(
\underline{P}\right)  $ (see also \cite{cap_theorems:SankarLiang_Conf}).%
\begin{figure}
[ptb]
\begin{center}
\includegraphics[
height=1.535in,
width=3.3892in
]%
{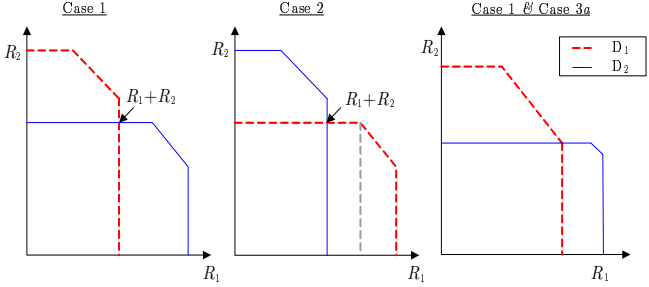}%
\caption{The rate region and sum-rate for cases $1$, $2$, and boundary case
$(1,3a)$.}%
\label{Fig_Case12}%
\end{center}
\end{figure}
\begin{figure}
[ptb]
\begin{center}
\includegraphics[
height=1.5082in,
width=3.2984in
]%
{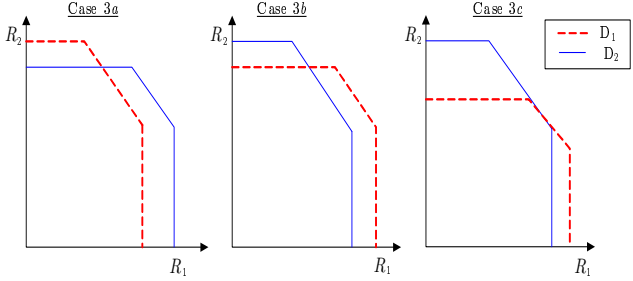}%
\caption{The rate region and sum-rate for cases $3a$, $3b$, and $3c$.}%
\label{Fig_Case3abc}%
\end{center}
\end{figure}

We broadly categorize the five possible choices for the sum-rate resulting
from the intersection of two pentagons into the sets of \textit{active} and
\textit{inactive} cases. The \textit{inactive set},\textit{ }consisting of
cases $1$ and $2$, includes all intersections of $\mathcal{C}_{1}\left(
\underline{P}\right)  $ and $\mathcal{C}_{2}\left(  \underline{P}\right)  $
for which the constraints on the two sum-rates are not active, i.e., no rate
tuple on the sum-rate plane achieved at one of the receivers lies within or on
the boundary of the rate region achieved at the other receiver. On the other
hand, the intersections for which there exists at least one such rate tuple
such that the two sum-rates constraints are active belong to the
\textit{active set}. This includes cases $3a$, $3b$, and $3c$ shown in Fig.
\ref{Fig_Case12} where the sum-rate at $D_{1}$ is smaller, larger, or equal,
respectively, to that achieved at $D_{2}$. By definition, the active set also
include the \textit{boundary cases} where there is exactly one such rate
pair.\ However, to simplify the optimization problem, we consider the six
boundary cases separately and denote them as cases $\left(  l,n\right)  $,
$l=1,2,$ and $n=3a,3b,3c$. We write $\mathcal{B}_{i}\subseteq\mathcal{P}$ and
$\mathcal{B}_{l,n}\subseteq\mathcal{P}$ to denote the set of power policies
that achieve case $i$, $i=1,2,3a,3b,3c$ and case $\left(  l,n\right)  $,
$l=1,2$, $n=3a,3b,3c$, respectively. Observe that cases $1$ and $2$ do not
share a boundary since such a transition (see Fig. \ref{Fig_Case12}) requires
passing through case $3a$ or $3b$ or $3c$. Finally, note that Fig.
\ref{Fig_Case3abc} illustrates two specific $\mathcal{C}_{1}$ and
$\mathcal{C}_{2}$ regions for $3a$, $3b$, and $3c$.

The occurrence of any one of the disjoint cases depends on both the channel
statistics and the policy $\underline{P}$. Since it is not straightforward to
know \textit{a priori} the power allocations that achieve a certain case, we
maximize the sum-capacity for each case over all allocations in $\mathcal{P}$
and write $\underline{P}^{(i)}$ and $\underline{P}^{(l,n)}$ to denote the
optimal solution for case $i$ and case $(l,n)$, respectively. Explicitly
including boundary cases ensures that the sets $\mathcal{B}_{i}$ and
$\mathcal{B}_{l,n}$ are disjoint for all $i$ and $(l,n)$, i.e., these sets are
either open or half-open sets such that no two of them share a boundary (see
\cite{cap_theorems:SankarLiang_Conf}). This in turn simplifies the convex
optimization as follows. Let $\underline{P}^{(i)}$ be the optimal policy
maximizing the sum-rate for case $i$ over all $\underline{P}\in\mathcal{P}$.
The optimal $\underline{P}^{(i)}$ must satisfy the conditions for case $i$,
i.e., $\underline{P}^{(i)}$ $\in$ $\mathcal{B}_{i}$. If the conditions are
satisfied, we prove the optimality of \underline{$P$}$^{\left(  i\right)  }$
using the fact that the rate functions for each case are concave. On the other
hand, when \underline{$P$}$^{\left(  i\right)  }$ $\not \in $ $\mathcal{B}%
_{i}$, it can be shown that $R_{1}+R_{2}$ achieves its maximum outside
$\mathcal{B}_{i}$. The proof again follows from the fact that $R_{1}+R_{2}$
for all cases is a concave function of \underline{$P$} for all $\underline{P}$
$\in$ $\mathcal{P}$. Thus, when \underline{$P$}$^{\left(  i\right)  }$
$\not \in $ $\mathcal{B}_{i}$, for every \underline{$P$} $\in$ $\mathcal{B}%
_{i}$ there exists a \underline{$P$}$^{\prime}$ $\in$ $\mathcal{B}_{i}$ with a
larger sum-rate. Combining this with the fact that the sum-rate expressions
are continuous while transitioning from one case to another at the boundary of
the open set $\mathcal{B}_{i}$, ensures that the maximum sum-rate is achieved
by some $\underline{P}$ $\not \in $ $\mathcal{B}_{i}$. Similar arguments
justify maximizing the optimal policy for each case over all $\mathcal{P}$.

The following theorem summarizes the optimal power policy for each case. The
optimal $\underline{P}^{\left(  i\right)  }$ or $\underline{P}^{\left(
l,n\right)  }$ maximizing the sum-rate for case $i$ or $\left(  l,n\right)  $
satisfies the conditions for only that case and is determined using Lagrange
multipliers and the \textit{Karush-Kuhn-Tucker} (KKT) conditions.

\begin{theorem}
\label{Th_1}A policy $\underline{P}^{\left(  i\right)  }$ or $\underline
{P}^{\left(  l,n\right)  }$ maximizes the sum-rate for case $i$,
$i=1,2,3a,3b,3c,$ or the boundary case $\left(  l,n\right)  $, $l=1,2$, and
$n=3a,3b,3c,$ when the entries \underline{$P$}$_{1}^{\left(  \cdot\right)  }$
and \underline{$P$}$_{2}^{\left(  \cdot\right)  }$ of $\underline{P}^{\left(
\cdot\right)  }$ satisfy
\begin{equation}%
\begin{array}
[c]{cc}%
f_{k}^{\left(  \cdot\right)  }\leq\nu_{k}\ln2 & k=1,2
\end{array}
\label{fk_bound}%
\end{equation}
where $\nu_{k}$ is chosen to satisfy (\ref{GIC_pwr_constraint}) such that%
\begin{equation}%
\begin{array}
[c]{ll}%
f_{k}^{(i)}=\frac{\left\vert h_{m,k}\right\vert ^{2}}{\left(  1+\left\vert
h_{m,k}\right\vert ^{2}P_{k}\right)  } &
\begin{array}
[c]{c}%
i=1:\left(  m,k\right)  =(1,1),(2,2)\\
i=2:\left(  m,k\right)  =(1,2),(2,1)
\end{array}
\\
f_{k}^{(i)}=\frac{\left\vert h_{m,k}\right\vert ^{2}}{\left(  1+\sum_{j=1}%
^{2}\left\vert h_{m,j}\right\vert ^{2}P_{j}\right)  } &
\begin{array}
[c]{c}%
i=3a:m=1\\
i=3b:m=2
\end{array}
\end{array}
\label{fk_defn}%
\end{equation}%
\begin{equation}%
\begin{array}
[c]{ll}%
\begin{array}
[c]{l}%
f_{k}^{\left(  i\right)  }=\left(  1-\alpha\right)  f_{k}^{\left(  3a\right)
}\\
~\ \ \ \ \ \ \ \ \ +\alpha f_{k}^{\left(  3b\right)  }%
\end{array}
& \text{ }i=3c\\%
\begin{array}
[c]{l}%
f_{k}^{\left(  l,n\right)  }=\left(  1-\alpha\right)  f_{k}^{\left(  n\right)
}\\
\text{ \ \ \ \ \ \ \ \ \ \ }+\alpha f_{k}^{\left(  l\right)  }\text{
\ \ \ \ \ \ }%
\end{array}
& \left(  l,n\right)  =(1,3a),(2,3b)\\%
\begin{array}
[c]{l}%
f_{k}^{\left(  l,3c\right)  }=\alpha_{3}f_{k}^{\left(  3a\right)  }\\
\text{ \ \ \ \ \ \ \ \ \ }+\alpha_{2}f_{k}^{\left(  3b\right)  }+\alpha
_{1}f_{k}^{(l)}%
\end{array}
& l=1,2
\end{array}
\end{equation}
with $\alpha$, $\alpha_{1}$, $\alpha_{2}$, and $\alpha_{3}=1-\alpha_{1}%
-\alpha_{2}$ chosen to satisfy the appropriate boundary conditions. The
optimal $\underline{P}^{\left(  i\right)  }\in\mathcal{B}_{i}$ or
$\underline{P}^{\left(  l,n\right)  }\in\mathcal{B}_{l,n}$ satisfies the
condition for case $i$ or case $(l,n)$, respectively. The conditions for each
case are given as%
\begin{align}
&
\begin{array}
[c]{cc}%
\text{Case }1: &
\begin{array}
[c]{c}%
I(X_{1};Y_{1}|X_{2}\mathbf{H})<I(X_{1};Y_{2}|\mathbf{H})\\
I(X_{2};Y_{2}|X_{1}\mathbf{H})<I(X_{2};Y_{1}|\mathbf{H})
\end{array}
\end{array}
\label{Case1_Cond}\\
&
\begin{array}
[c]{cc}%
\text{Case }2: &
\begin{array}
[c]{c}%
I(X_{1};Y_{1}|X_{2}\mathbf{H})<I(X_{1};Y_{2}|\mathbf{H})\\
I(X_{2};Y_{2}|X_{1}\mathbf{H})<I(X_{2};Y_{1}|\mathbf{H})
\end{array}
\end{array}
\label{Case2_Cond}%
\end{align}%
\begin{align}
&
\begin{array}
[c]{cc}%
\text{Case }3a: & I(X_{1}X_{2};Y_{1}|\mathbf{H})<I(X_{1}X_{2};Y_{2}%
|\mathbf{H})
\end{array}
\label{Case3a_Cond}\\
&
\begin{array}
[c]{cc}%
\text{Case }3b: & I(X_{1}X_{2};Y_{1}|\mathbf{H})>I(X_{1}X_{2};Y_{2}%
|\mathbf{H})
\end{array}
\label{Case3b_Cond}\\
&
\begin{array}
[c]{cc}%
\text{Case }3c: & I(X_{1}X_{2};Y_{1}|\mathbf{H})=I(X_{1}X_{2};Y_{2}%
|\mathbf{H})
\end{array}
\label{Case3c_Cond}%
\end{align}%
\begin{equation}%
\begin{array}
[c]{cc}%
\text{Case }\left(  l,n\right)  : &
\begin{array}
[c]{c}%
\text{Satisfy cases }n\text{ }\&\text{ }l\text{ with }\\
\text{equality for one case }l\text{ condition }%
\end{array}
\end{array}
\label{CaseBC_Cond}%
\end{equation}
where in (\ref{Case1_Cond})-(\ref{CaseBC_Cond}), $X_{1}$ and $X_{2}$ are
Gaussian distributed subject to (\ref{GIC_pwr_constraint}). The $\underline
{P}^{\ast}$ that maximizes the sum-capacity is obtained by computing
$\underline{P}^{(i)}$ or $\underline{P}^{(l,n)}$ starting with the inactive
cases, followed by the boundary cases $(l,n)$, and finally the active cases
$3a,$ $3b,$ and $3c$ until for some case the corresponding $\underline
{P}^{(i)}$ or $\underline{P}^{(l,n)}$ satisfies the case conditions.
\end{theorem}

From (\ref{fk_bound}) and (\ref{fk_defn}), one can easily verify that for the
inactive cases $1$ and $2$ the optimal policies involve the classic
water-filling solution over point-to-point links$.$ Specifically, the optimal
policies for cases $1$ and $2$ simplify to water-filling over the two
bottle-neck links $(S_{1}\rightarrow D_{1})$, $(S_{2}\rightarrow D_{2})$ and
$(S_{1}\rightarrow D_{2})$, $(S_{2}\rightarrow D_{1})$, respectively. On the
other hand, for the active cases $3a$ and $3b$, the optimal allocation at each
source simplifies to the opportunistic water-filling allocation for a MAC
\cite{cap_theorems:Knopp_Humblet,cap_theorems:TH01} such that in each channel
use the source with the larger $f_{k}^{\left(  i\right)  }/\nu_{k}$ for case
$i,$ $i=3a,3b$, transmits. Observe that the water-filling solutions are with
respect to the receiver that achieves the smaller sum-capacity. Finally, for
all the boundary cases including case $3c$, the optimal policy for source $k$
is still an opportunistic solution such that the source with the larger
$f_{k}^{(i)}/\nu_{k}$ or $f_{k}^{(l,n)}/\nu_{k}$, $i=3c$ and for all $(l,n)$,
transmits. However, unlike the other cases, the optimal policy at each source
for the boundary cases is no longer a water-filling solution; instead for each
channel instantiation the optimal policy at source $k$ satisfies
(\ref{fk_bound}) with equality when the users are opportunistically scheduled.

The conditions in (\ref{Case1_Cond}) and (\ref{Case2_Cond}) for the two
inactive cases exclude all other cases and define the disjoint sets
$\mathcal{B}_{1}$ and $\mathcal{B}_{2}$. Similarly, the conditions for the six
boundary cases define the disjoint sets $\mathcal{B}_{l,n}$ for all $\left(
l,n\right)  $. However, the conditions for $3a$, $3b$, and $3c$ can be
satisfied by the boundary cases. To ensure that the sets $\mathcal{B}_{3a}$,
$\mathcal{B}_{3b}$, and $\mathcal{B}_{3c}$ are disjoint from all other sets,
the algorithm for determining the optimal \underline{$P$}$^{\ast}$ requires
eliminating a case at a time starting from case $1$. Thus, the algorithm first
eliminates the inactive cases, and then checks for the boundary cases, and
finally checks for cases $3a,$ $3b$, and $3c$.

\begin{remark}
The capacity region, $\mathcal{C}_{\text{C-MAC}}$ can be completely
characterized by using the same approach to maximize the sum $\mu_{1}R_{1}%
+\mu_{2}R_{2}$, for all $\left(  \mu_{1},\mu_{2}\right)  $ pairs. In general,
each tuple on the boundary of $\mathcal{C}_{\text{C-MAC}}$ may be maximized by
a different case, and thus, the optimal policy is also a function of $\left(
\mu_{1},\mu_{2}\right)  $.
\end{remark}

\section{\label{Sec_4}IFC:\ Converse}

We now apply the results in Theorem \ref{Th_1} to the ergodic fading IFC. For
the IFC, the power policies satisfying (\ref{fk_bound}) and (\ref{fk_defn})
are achievable when $D_{1}$ and $D_{2}$ decode messages from both sources; the
resulting C-MAC sum-capacity is a lower bound on the IFC sum-capacity. Below,
we present a converse to show that these sum-rate lower bounds are tight for
the classes of strong and very strong ergodic fading IFC. The convex capacity
region of the ergodic fading two-user IFC, $\mathcal{C}_{\text{IFC}}$, can be
bounded by hyperplanes $C\left(  \mu_{1},\mu_{2}\right)  $ such that for all
$\mu_{1}>0$ and $\mu_{2}>0$, we have%
\begin{equation}
\mathcal{C}_{\text{IFC}}=\left\{  \left(  R_{1},R_{2}\right)  :\mu_{1}%
R_{1}+\mu_{2}R_{2}\leq C\left(  \mu_{1},\mu_{2}\right)  \right\}
\end{equation}
subject to (\ref{GIC_pwr_constraint}). The boundary of $\mathcal{C}%
_{\text{IFC}}$ is determined by maximizing $\mu_{1}R_{1}+\mu_{2}R_{2}$ for
each choice of $(\mu_{1},\mu_{2})$ over all $\underline{P}\in\mathcal{P}$. For
the sum-capacity, we set $\mu_{1}=\mu_{2}=1$.

\subsection{Very Strong Ergodic IFC}

\begin{definition}
A very strong ergodic fading IFC with respect to the tuple $\left(  \mu
_{1},\mu_{2}\right)  $ results when a $\underline{P}\in\mathcal{P}$ and
$\mathbf{H}$ satisfy
\begin{equation}%
\begin{array}
[c]{c}%
I(X_{1};Y_{1}|X_{2}\mathbf{H})<I(X_{1};Y_{2}|\mathbf{H})\\
I(X_{2};Y_{2}|X_{1}\mathbf{H})<I(X_{2};Y_{1}|\mathbf{H})
\end{array}
\label{VSIFC_Cond}%
\end{equation}
for all choices of $X_{1}$ and $X_{2}$.
\end{definition}

\begin{theorem}
\label{Th_VS}The sum-capacity of a class of very strong ergodic Gaussian IFCs
is
\begin{equation}
\sum_{k=1}^{2}\mathbb{E}\left[  \log\left(  1+\left\vert H_{k,k}\right\vert
^{2}P_{k}^{\ast}\left(  \mathbf{H}\right)  \right)  \right]
\label{VIFC_Con_SC}%
\end{equation}
where \underline{$P$}$_{k}^{\ast}\left(  \mathbf{H}\right)  $, for all $k$, is
the optimal water-filling solution for single-sender single-receiver ergodic
fading links.
\end{theorem}

\begin{proof}
An outer bound on the sum-capacity of the IFC can be obtained by setting
$H_{j,k}=0$ for all $j\not =k$, i.e., by assuming no interference. In the
absence of interference, Gaussian signaling achieves capacity for each of the
$S_{k}$ to $D_{k}$ links, $k=1,2$, and the resulting sum-capacity is given by
(\ref{VIFC_Con_SC}) where, \underline{$P$}$_{k}^{\ast}\left(  \mathbf{H}%
\right)  $ is the optimal water-filling solution for single-sender
single-receiver ergodic fading links, i.e., it satisfies the condition in
(\ref{fk_bound}) for $f_{k}^{(1)}$ in (\ref{fk_defn}), subject to
(\ref{GIC_pwr_constraint}). From Theorem \ref{Th_1}, we see that when the
channel statistics and the power policy satisfy (\ref{Case1_Cond}), i.e.,
$\underline{P}^{\ast}=\underline{P}^{(1)}\in\mathcal{B}_{1}$, the achievable
strategy of decoding both messages at both destinations achieves this
sum-capacity outer bound. Thus, the sum-capacity of a very strong IFC is that
of a C-MAC for which $\underline{P}$ satisfies case $1$ conditions, i.e.,
$\underline{P}$ satisfies (\ref{VSIFC_Cond}).
\end{proof}

For a deterministic $\mathbf{H}$, the conditions in (\ref{VSIFC_Cond})
simplify to those for the very strong non-fading IFC in
\cite{cap_theorems:Sato_IC}. Further, from Fig. \ref{Fig_Case12}, we see that
as with the non-fading very strong IFC, the intersecting region for Case $1$
is also a rectangle; note, however that unlike the non-fading case, this
rectangle is not the entire capacity region but only the region achieving the
sum-capacity. Finally, note that the condition in (\ref{VSIFC_Cond}) depends
on both the channel statistics and the transmit power.

\begin{remark}
In contrast, the conditions for case $2$ in (\ref{Case2_Cond}) model a
\textit{weak ergodic }IFC for which the C-MAC sum-capacity is strictly a lower bound.
\end{remark}

\subsection{Strong Ergodic IFC}

\begin{definition}
A strong ergodic fading IFC with respect to the tuple $\left(  \mu_{1},\mu
_{2}\right)  $ results when a $\underline{P}\in\mathcal{P}$ and $\mathbf{H}$
satisfies%
\begin{align}
I(X_{1};Y_{1}|X_{2}\mathbf{H})  &  <I(X_{1};Y_{2}|X_{2}\mathbf{H}%
)\label{SIC_Cond1}\\
I(X_{2};Y_{2}|X_{1}\mathbf{H})  &  <I(X_{2};Y_{1}|X_{1}\mathbf{H})
\label{SIC_Cond2}%
\end{align}
for all choices of $X_{1}$ and $X_{2}$.
\end{definition}

\begin{theorem}
\label{Th_S}The sum-capacity of the class of strong ergodic fading Gaussian
IFCs is
\begin{equation}
\min_{j=1,2}\left\{  \mathbb{E}\left[  C\left(  \sum\nolimits_{k=1}%
^{2}\left\vert H_{j,k}\right\vert ^{2}P_{k}^{\ast}\right)  \right]  \right\}
\end{equation}
where, for all $k$, $P_{k}^{\ast}=P_{k}^{\left(  i\right)  }$ or $P_{k}^{\ast
}=P_{k}^{\left(  l,n\right)  }$ for $l=1$ and $i,n\in\left\{
3a,3b,3c\right\}  $.
\end{theorem}

\begin{proof}
Due to lack of space, we present a proof sketch. We use the fact that the
channel states are independent of the source messages, Fano's and the data
processing inequality, the ergodicity of the channel for large $N$, the fact
that $\underline{P}\in\mathcal{P}$ satisfies (\ref{SIC_Cond1}) and
(\ref{SIC_Cond2}), and the optimality of Gaussian signaling to upper bound the
sum-rate as
\begin{align}
R_{1}+R_{2}  &  \leq\left[  I(X_{1};Y_{1}|X_{2}\mathbf{H})+I(X_{2}%
;Y_{2}|\mathbf{H})\right] \label{SIFC_RsumB1}\\
&  \leq\left[  I(X_{1};Y_{2}|X_{2}\mathbf{H})+I(X_{2};Y_{2}|\mathbf{H})\right]
\label{SIFC_C1}\\
&  \leq I(X_{1}X_{2};Y_{2}|\mathbf{H})\label{SIFC_C2}\\
&  \leq\mathbb{E}\left[  C\left(  \sum\nolimits_{k=1}^{2}\left\vert
H_{2,k}\right\vert ^{2}P_{k}^{\ast}\right)  \right]  \label{SIFC_C4}%
\end{align}
One can similarly show that
\begin{equation}
R_{1}+R_{2}\leq\mathbb{E}\left[  C\left(  \sum\nolimits_{k=1}^{2}\left\vert
H_{1,k}\right\vert ^{2}P_{k}^{\ast}\right)  \right]  ,
\end{equation}
and thus, from (\ref{CMAC_Cj}) we see that the sum-rate is upper bounded by
the sum-capacity of a C-MAC. Further, we can bound $R_{1}\leq\mathbb{E[}%
C(\left\vert H_{1,k}\right\vert ^{2}P_{k}^{\ast})]$ and $R_{2}\leq$
$\mathbb{E[}C(\left\vert H_{2,k}\right\vert ^{2}P_{k}^{\ast})]$, and thus,
from Corollary \ref{Cor_1}, the sum-capacity of a the ergodic C-MAC
sum-capacity is also the sum-capacity of the ergodic IFC when (\ref{SIC_Cond1}%
) and (\ref{SIC_Cond2}) hold.

\textit{Optimal Power Allocation }: From Theorem \ref{Th_1}, the optimal
$\underline{P}^{\ast}$ for an ergodic C-MAC satisfies only one of the
conditions in (\ref{Case1_Cond})-(\ref{CaseBC_Cond}). Further, from Theorem
\ref{Th_1} and Figs. \ref{Fig_Case12} and \ref{Fig_Case3abc}, the conditions
in (\ref{SIC_Cond1}) and (\ref{SIC_Cond2}) can be satisfied by $7$ different
cases, namely, cases $1$, $3a$, $3b$, $3c$, $\left(  1,3a\right)  $, $\left(
1,3b\right)  $, and $\left(  1,3c\right)  $. Since these cases are mutually
exclusive, the optimal $\underline{P}^{\ast}$ is given by that optimal policy
which in addition to satisfying the condition for one of the above listed
cases also satisfies (\ref{SIC_Cond1}) and (\ref{SIC_Cond2}). For example,
suppose $\underline{P}^{\ast}$ satisfies the condition for case $1$, i.e.,
$\underline{P}^{\ast}\in\mathcal{B}_{1}$. Since the conditions for this very
strong case in (\ref{Case1_Cond}) (see also (\ref{VSIFC_Cond})) imply the
conditions for the strong case in (\ref{SIC_Cond1}) and (\ref{SIC_Cond2}), the
sum-capacity and the optimal power policy are directly given by Theorem
\ref{Th_VS}. On the other hand, suppose $\underline{P}^{\ast}\in
\mathcal{B}_{3a}$, i.e., $\underline{P}^{\ast}$ satisfies the conditions in
Theorem \ref{Th_1} only for case $3a$ (see Fig. \ref{Fig_Case3abc}). Since
$\underline{P}^{\ast}$ also satisfies (\ref{SIC_Cond1}) and (\ref{SIC_Cond2}),
we define an open (or half-open) set $\mathcal{B}_{3a}^{\prime}\subset
\mathcal{B}_{3a}$ in which (\ref{SIC_Cond1}), (\ref{SIC_Cond2}), and
(\ref{Case3a_Cond}) are all satisfied. The concavity of the sum-rate
expression for this case then guarantees that the optimal policy is unique and
belongs to the open set $\mathcal{B}_{3a}^{\prime}$ (see also the arguments in
Section \ref{section 3}). Note that the requirement that $\underline{P}%
^{(3a)}$ satisfy (\ref{SIC_Cond1}) and (\ref{SIC_Cond2}) only limits
$\underline{P}^{(3a)}$ to $\mathcal{B}_{3a}^{\prime}$ and does not change the
solution presented in Theorem \ref{Th_1} for this case. Thus, the sum-capacity
for this case is
\begin{equation}
\mathbb{E}\log\left(  1+\sum_{k=1}^{2}\left\vert H_{1,k}\right\vert ^{2}%
P_{k}^{(3a)}\left(  \mathbf{H}\right)  \right)  .
\end{equation}
The arguments above also apply to the remaining cases listed above.
\end{proof}

\begin{remark}
The conditions in (\ref{SIC_Cond1}) and (\ref{SIC_Cond2}) are ergodic
generalizations of the conditions presented in \cite{cap_theorems:Sato_IC} for
the non-fading strong IFC (see also \cite[(1),(2)]%
{cap_theorems:CostaElGamal_IC} for the discrete memoryless strong IFC).
However, unlike the non-fading Gaussian IFC, the conditions in
(\ref{SIC_Cond1}) and (\ref{SIC_Cond2}) for the ergodic Gaussian\ IFC depend
on both the channel statistics and the power policy $\underline{P}$. Further,
as expected, the very strong ergodic IFC is a special case of the strong IFC
where (\ref{VSIFC_Cond}) holds.
\end{remark}

\bibliographystyle{IEEEtran}
\bibliography{IC_refs}

\bigskip
\end{document}